\documentclass[twocolumn,showpacs,amsmath,amssymb,prb,superscriptaddress]{revtex4-1}

\usepackage[english]{babel}
\usepackage{graphicx}
\allowdisplaybreaks 
\usepackage{xcolor}
 \usepackage{times}
\usepackage{amsthm}

\usepackage[colorlinks,bookmarks=true,citecolor=blue,linkcolor=red,urlcolor=blue]{hyperref}

	\newtheorem{lemma}{Lemma}
	\newcommand{\degga}{^{\dag}}		
	\newcommand{\ket}[1]{\left| #1 \right>}	
	\newcommand{\bra}[1]{\left< #1 \right|}	
	\newcommand{\braket}[2]{\left< #1 \vphantom{#2} \right|\left. #2 \vphantom{#1} \right>}				
	\newcommand{\matrixelm}[3]{\left< #1 \vphantom{#2#3} \right| #2 \left| #3 \vphantom{#1#2} \right>} 	
	\newcommand{\abs}[1]{\left| #1 \right|}	
	\newcommand{\norm}[1]{\left|\left| #1 \right|\right|}	
	\DeclareMathOperator{\tr}{Tr}
	\DeclareMathOperator{\id}{{\mathbb{I}}}
	\newcommand{\av}[1]{\left< #1\right>}	
	\newcommand{\bb}[1]{\mathbb{#1}}		%
	\newcommand{\mcal}[1]{\mathcal{#1}}		%
	\newcommand{\cc}{\mathbb{C}}

 	\renewcommand{\vec}[1]{\boldsymbol{ #1 }}

\begin{document}
\title{Composite symmetry protected topological order and effective models}
\author{A. Nietner}
\affiliation{Dahlem Center for Complex Quantum Systems, Freie Universit\"at Berlin, D-14195 Berlin, Germany}
\author{C. Krumnow}
\affiliation{Dahlem Center for Complex Quantum Systems, Freie Universit\"at Berlin, D-14195 Berlin, Germany}
\author{E. J. Bergholtz}
\affiliation{Department of Physics, Stockholm University, AlbaNova University Center, 106 91 Stockholm, Sweden}
\affiliation{Dahlem Center for Complex Quantum Systems, Freie Universit\"at Berlin, D-14195 Berlin, Germany}
\author{J. Eisert}
\affiliation{Dahlem Center for Complex Quantum Systems, Freie Universit\"at Berlin, D-14195 Berlin, Germany}
\date{\today}

\begin{abstract}
Strongly correlated quantum many-body systems at low dimension exhibit a wealth of phenomena, ranging from features of geometric frustration to signatures of symmetry-protected topological order. In suitable descriptions of such systems, it can be helpful to resort to effective models which focus on the essential degrees of freedom of the given model. In this work, we analyze how to determine the validity of an effective model by demanding it to be in the same phase as the original model. We focus our study on one-dimensional spin-$1/2$ systems and explain how non-trivial 
symmetry protected topologically ordered (SPT) phases of an effective spin $1$ model can arise depending on the couplings in the original Hamiltonian. In this analysis, tensor network methods feature in two ways: On the one hand, we make use of recent techniques for the classification of SPT phases using matrix product states in order to identify the phases in the effective model with those in the underlying physical system, employing K\"unneth's theorem for cohomology. As an intuitive paradigmatic model we exemplify the developed methodology by investigating the bi-layered $\Delta$-chain. For strong ferromagnetic inter-layer couplings, we find the system to transit into exactly the same phase as an effective spin $1$ model. However, for weak but finite coupling strength, we identify a symmetry broken phase differing from this effective spin-$1$ description. On the other hand, we underpin our argument with a numerical analysis making use of matrix product states.

\end{abstract}

\maketitle
\section{Introduction}

Strongly correlated quantum many-body systems exhibit a wide range of intriguing properties, such as 
long range magnetic ordering \cite{MagneticOrder,MagneticExcitations}  
or notions of intrinsic or symmetry-protected topological order  \cite{WenBook,WenClassification,EntanglementSpectrumTopologicalPhase,PEPSTopology}.
In systems exhibiting geometric frustration the phenomenology is enriched by effects such as, spin-ice behavior, the emergence of magnetic monopoles, and 
fractionalisation coming into play \cite{MoessnerSpinIce,PyrochloreOxides}.
Such strongly correlated and geometrically frustrated systems enjoy a significant experimental interest, and are probed with techniques such as neutron diffraction allowing for a window into their physics \cite{NeutronScattering,BellaCriticality}   
At the same time, the theoretical and numerical description of these systems constitute significant challenges,
at least in more than one spatial dimension.
Even powerful numerical methods have difficulties to identify the essential ground state features in such systems accurately, as can be seen on the example of the race to identify the ground state properties of the anti-ferromagnetic spin-1/2 Heisenberg model on the Kagom\'e lattice
\cite{WhiteKagome,KagomeSchollwoeck,Chen}. 
In layered materials, properties of such models become even more complex as the system can favor the creation of compounds. Here, 
behavior
such as the formation of a spin liquid based on a novel frustration mechanism can appear \cite{SpinLiquid}.

In general, the description of such interacting quantum many-body systems is often simplified by the use of \emph{effective models} in which the description is reduced to the essential degrees of freedoms of the system needed in order to capture for instance the ground state properties.
In order to be valid, the effective model should of course resemble the essential physics of the original one.
In the most rough reading of this requirement, both systems should be in the same phase.

In this work, we investigate the problem of how to validate an effective model in the tractable setting of one-dimensional spin-$1/2$ systems in order to highlight the essential features of the method. In such systems, effective models are for instance derived by grouping several spins into a single site of definite higher spin.
In the easiest instance 
we think of a ladder of spin-$1/2$ spins which have a horizontal ferromagnetic coupling of pairs and an arbitrary vertical coupling. 
Intuitively, for strong ferromagnetic couplings the pair of $1/2$ spins form a compound of spin $1$ for low energies such that an effective spin $1$ theory, i.e.,~neglecting the spin $0$ sector, should allow to capture the ground state properties of this system accurately.
In such a setting, we require the blocked model and the effective model to be in the same \emph{symmetry protected topologically ordered} (SPT) phase --- ensuring that the most essential features of the system are accounted for. While we put emphasis on a paradigmatic situation to be specific, the overall approach pursued here can be applied to more elaborate settings.

Specifically, in this work we establish a connecting between the SPT phases of the different models in order to validate the use of an effective model in general.
In particular, we show how to obtain from a $SU(2)$ symmetric spin-$1/2$ system (which is always in a trivial phase) a system in a non-trivial SPT phases by composing sites.
This transition roots on the blocking of two spin-$1/2$ spins into an indivisible unit cell. 
The symmetries of the resulting blocked system are then identified with the ones of the effective spin $1$ model which enables us to compare the corresponding SPT phase of the system.
For this, we rely on well established tools for the classification of SPT phases using \emph{matrix product states (MPS)}\cite{WenClassification,ClassificationPhases,EntanglementSpectrumTopologicalPhase} and the cohomology of the respective symmetry groups. 
Using the identification of the symmetry groups of the blocked and effective model, we then derive an order parameter for the detection of this SPT phase of the effective model on the level of the underlying spin-$1/2$ system giving a rigorous tool to compare the effective spin $1$ physics with the physics of the original model.

Further, we illustrate the approach on a geometrically frustrated one-dimensional system --- the bi-layered $\Delta$ chain (BDC) (cf.~Fig.~\ref{lattice}) with an inter-layer ferromagnetic coupling and an anti-ferromagnetic coupling inside each layer. 
This system serves as a one-dimensional archetype of geometrically frustrated systems and appears in the effective description of more complex 
higher dimensional layered systems \cite{SpinLiquid, Kikuchi2011}. 
For the BDC, we investigate the validity of an effective spin-$1$ description, for which the ground state is numerically found to be in the 
\emph{Haldane phase}, depending on the strength of the inter-layer coupling using perturbation theory and numerical \emph{tensor network} (TN) methods\cite{VerstraeteBig,OrusReview,TensorNetworkReview,AreaReview,MPSRev}.
By this argument, we analytically find that starting at a finite coupling strength the blocked system can be well described using the effective spin-$1$ theory where the exact critical coupling strength is determined numerically to be $J_F \approx 0.66 J_{AF}$ with $J_F$ and $J_{AF}$ denoting the corresponding ferromagnetic and anti-ferromagnetic coupling strength.
Finally, applying our order parameter we align our results with studies on the spin-$1/2$ Heisenberg ladder\cite{WhiteSpin1Chain, Pollmann2012, Dagotto1996, Totsuka1995, Matsumoto2004}.

This work is structured into two distinct parts.
We first review the classification of SPT phases in one-dimensional systems using matrix product states and explain in mathematical terms
how non-trivial SPT phases emerge from blocking sites on the example of spin-$1/2$ systems and establish a connection between the order parameters of the blocked and the effective model.  
In the second part, we discuss the concrete situation of the BDC as an example and apply the developed methods.


\section{Symmetry protected topological order}
We begin this section by reviewing the classification of SPT phases using MPS. 
We then argue that the symmetries of a blocked spin-$1/2$ system agree with the ones of an effective spin-$1$ theory and establish an order parameter that can be used to check the validity of the effective model.

\subsection{Group cohomology and SPT in 1D}
Much of what follows will build on the connection between \emph{symmetry protected topological order} in one-dimensional
systems and the second cohomology class of the respective symmetry. In this subsection, we will briefly review some notions 
made use of later in a language of MPS \cite{DMRGWhite92,FCS,MPS}. 
As is well-known, for such one-dimensional systems, matrix product states approximate ground states arbitrarily well, 
at the expense of a bond-dimension that scales moderately
with the system size \cite{SchuchApprox}. This is essentially rooted in the observation of ground states satisfying an \emph{area law}
for suitable entanglement entropies for gapped models \cite{AreaLawOneD,AreaReview}.
In such a language of TN states,
notions of topological order can be particularly concise and at the same time rigorously captured \cite{WenBook,TopologicalOrderInPEPS,WenClassification,ClassificationPhases, EntanglementSpectrumTopologicalPhase,ResonatingValenceBondStates,PEPSTopology,TopologicalOrderInPEPS}. 
In this mindset and in this formalism, an emphasis is put on ground states, while local Hamiltonians reenter stage by means of the concept of a \emph{parent Hamiltonian}.

To be specific, consider a \emph{uniform matrix product state} (uMPS) vector\cite{2011PhRvL.107g0601H}  
in canonical form $\ket{\psi}$ parametrized by the tensor $M_{\alpha,\beta}^i$, where $\alpha,\beta=1,\dots, D$
reflect the virtual indices with bond dimension $D$ and $i$ runs from $1$ to the local physical dimension. When the state vector is 
symmetric under a local transformation $g\in G$
for some symmetry group $G$, in that $(U_g)^{\otimes n}\ket{\psi} = \ket{\psi}$,
together with a transformation of the matrix $M^i\mapsto M^{i\bullet}$ with $\bullet$ 
indicating either complex conjugation or transposition, then it is shown 
in Ref.~\cite{SymmetriesMPS} 
that $U_g$ is reflected on the virtual level as
$U_g \simeq V_g \otimes  \bar V_g$, where $V_g$ is a projective representation of $G$,
so that $V_g$ are unitary and uniquely defined up to a phase. Let now $\ket{\psi}$ be symmetric with respect to the group $G$. Then, for any $g\in G$ we obtain through this formula a $V_g$ which form a representation of $G$ on the virtual level. Fixing the $U(1)$ gauge freedom in the virtual representatives and iterating this formula we find for $g,h\in G$
\begin{align}
	V_gV_h=\omega(g,h)V_{gh}.
	\label{cocycleCommutation}
\end{align}
Using Eq.~\eqref{cocycleCommutation} iteratively on $V_kV_gV_h$ for any $k,g,h\in G$ one finds the relation
\begin{align}\label{cocycle}
	1=\frac{\omega(h,g)\omega(hg,k)}{\omega(h,gk)\omega(g,k)}.
\end{align}
Eq.~\eqref{cocycle} is the $2$-cocycle equation. Any such $\omega$ defines a projective representation of $G$ 
with elements $V_g$. Given a group $G$, the set of possible cocycles $\omega$ over the $G$-module $U(1)$ is not arbitrary but corresponds to $H^2(G,U(1))$, namely the 
\emph{second cohomology group of $G$ over $U(1)$}.
With this in mind, consider a family of MPS $\ket{\psi(\alpha)}$ parametrized by $\alpha\in[0,1]$ in a specific symmetry sector of the symmetry $G$.
In Refs.\ \cite{WenClassification,ClassificationPhases,EntanglementSpectrumTopologicalPhase} it is shown
that the cocycle corresponding to $\ket{\psi(\alpha)}$ can only change from $\alpha=0$ to $\alpha=1$ if the gap of the parent Hamiltonian corresponding to these MPS closes for an $0<\alpha_c<1$. Moreover, they show that any two MPS with the same corresponding cocycle $\omega$ can be connected by a smooth path along which the respective parent Hamiltonian remains gapped. This, however, means that the possible SPT phases with respect to $G$, which are characterized by symmetry protected long range entanglement\cite{Chen2010}, correspond to the non-trivial elements of its second cohomology group,
\begin{align}
	\text{SPT phases of } G\quad\leftrightarrow \quad H^2(G,U(1)).
\end{align}

\subsection{Effective symmetry group}
The specific question addressed in this work is how the SPT phases change when changing from the original model described by the Hamiltonian $H$ to a blocked model described by the Hamiltonian $H_b$, and how they relate to an effective model. 
Therefore, let us first introduce the notions of blocked and effective models and identify the collection of symmetries with respect to  which $H$ and its ground state are invariant. 

As explained in the introduction, we would like to focus our attention to one-dimensional spin-$1/2$ lattice systems. Depending on the interactions present in the Hamiltonian, such a system might be well described by an effective model that captures only the relevant degrees of freedom which we want to assume to originate from a blocked set of original sites. 
To be precise we block two spin-$1/2$ sites of the system into a new site with local Hilbert space $\cc^{2}\otimes \cc^{2}$ and consider them to be the natural unit of the system. The representation of the Hamiltonian with respect to this unit cell will be denoted by $H_b$. The effective model in our case is then the corresponding spin-$1$ system arising from the decomposition of the local Hilbert space into $\frac{1}{2}\otimes\frac{1}{2} = 0\oplus1$ and then neglecting the spin $0$ degrees of freedom.

We assume the systems under investigation to be invariant under $SU(2)$ transformations. 
Further those systems frequently favor additional discrete symmetries such as a time reversal (TR) or lattice inversion (I) symmetry. 
In layered models such as considered in the second part of this work also an additional layer exchange (LE) symmetry can be present.

In order to tell what happens due to the blocking it is important to observe how these symmetries transform under the blocking procedure.
Assume a bi-layered model with the symmetries introduced above and block two opposing vertices in each layer.
The TR symmetry remains and just changes its representation. The LE symmetry becomes a local $\mathbb{Z}_2$ unitary transformation and the I symmetry becomes an I symmetry with respect to  a vertex in the blocked model. 
The most important changes come from the local unitary $SU(2)$ transformations. 
Blocking a pair of vertices maps
\begin{align}
	SU(2) \rightarrow \mcal{G}=\{g\otimes g:\;g\in SU(2)\}
\end{align}
where we use the fundamental SU(2) representation on $\mathbb{C}^2$ (i.e., the spin-$1/2$ representation). In the following we will denote by $\mcal{G}_{tot}$ the full symmetry group under which $H_b$ is invariant. 

It is known that the second cohomology class of the $SU(2)$ is trivial, implying the fact that there are no corresponding non-trivial SPT phases for spin-$1/2$ systems. However, the new insight is that the respective symmetry group due to the blocking $\mcal{G}$ is isomorphic to $\mcal{G}\simeq SO(3)$ having non-trivial second cohomology $H^2(SO(3),U(1))\simeq\bb{Z}_2$ giving rise to two SPT phases --- a trivial phase and an ordered phase which we refer to as Haldane phase here. 
This essentially roots in the fact that by blocking two spin-$1/2$ sites one restricts the possible $SU(2)$ representations to the integer spin sector, locally corresponding to a faithful $SO(3)$ representation, 
as the total number of sites is thereby restricted to be even. 
In what follows, we make this isomorphism explicit and show using K\"unneth's theorem that the topological phases with respect to $\mcal{G}$ manifest themselves in the topological phases of $\mcal{G}_{tot}$.

\subsection{Isomorphism between $\mcal{G}$ and $SO(3)$}
\label{sec:IsomorphismOfGAndSO3}
In this section, we establish the isomorphism between $\mcal{G}$ and $SO(3)$. 
It essentially uses the fact that $\mathcal{G}\simeq SU(2)/\bb{Z}_2$ as the tensor product factors out 
$\mathbb{Z}_2$. Then we show and use the fact that $SU(2)/\mathbb{Z}_2$ is a projective orthogonal representation of $SO(3)$.
We lay out the details here as understanding this isomorphism allows for a clearer interpretation of the order parameter derived in the subsequent section.

\begin{lemma}[Group isomorphism]
There exists a group isomorphism $S^*:SU(2)/\mathbb{Z}_2\rightarrow SO(3)$.
\end{lemma}

\begin{proof}
To start with, it can be easily seen that the map $\rho:SU(2)/\bb{Z}_2\rightarrow\mcal{G}$ which is defined by $[g]\mapsto g\otimes g$ is a faithful representation (where $[g]$ denotes the equivalence class $g\bb{Z}_2\in SU(2)/\bb{Z}_2$ and the map is independent of the representative of that class).
Now, we use the fact that the map $\iota$ from the tensor product of matrices into the set of completely positive  maps 
\begin{align}
	\iota:\mcal{G}\rightarrow\mcal{G}^+=\{\Pi_g:=g(\cdot)g^\dagger|g\in SU(2)\},\quad g\otimes g\mapsto g(\cdot)g^\dagger
\end{align}
is an isomorphism if $\mcal{G}$ and $\mcal{G}^+$ are interpreted as groups. 
Next we construct a isomorphism from the latter to $SO(3)$. 
Consider therefore $\bb{H}_3:=\text{span}_{\bb{R}}\{\sigma_x,\sigma_y,\sigma_z\}$ the real space of Hermitian traceless complex $2\times2$ matrices, where the $\sigma_i$ are the Pauli matrices. 
It is natural to define a scalar product in $\bb{H}_3$ as 
\begin{align}
	\braket{a}{b}_{\bb{H}_3}=\alpha\tr[ab]
\end{align}
for any pre-factor $\alpha>0$ and any $a,b\in\bb{H}_3$.
Using the commutation relations of Pauli matrices 
\begin{align}
	\sigma_i\sigma_j=\delta_{i,j}+i\epsilon_{i,j,k}\sigma_k
\end{align}
it turns out that choosing $\alpha=\frac{1}{2}$ renders the map 
\begin{align}
	S:\bb{R}^3\rightarrow\bb{H}_3:\vec{v}\mapsto\vec{v}\cdot\vec{\sigma},
\end{align}
which is surjective, an isometry (preserving the scalar product) such that there exists an inverse $S^{-1}$. 
Let us now make full use of $S$ in order to connect $\mcal{G}^+$ and $SO(3)$. For any $\Pi_U\in\mcal{G}^+$ we define the map $R(\Pi_U):=S^{-1}\circ\Pi_U\circ S$ which can by definition be expressed as $R(\Pi_U):=S^{-1}\circ (U S(\cdot) U\degga)$ for some $U\in SU(2)$. Obviously $R(\Pi_U)$ is linear and we find (using the fact that $S$ is an isometry) 
\begin{align}
	&\braket{R(\Pi_U)\vec{v}}{R(\Pi_U)\vec{w}}_{\bb{R}^3}=\frac{1}{2}\tr[U\vec{v}\cdot\vec{\sigma}U\degga U\vec{w}\cdot\vec{\sigma}U\degga] \notag \\[8pt]
	&\quad=\frac{1}{2}\tr[\vec{v}\cdot\vec{\sigma}\vec{w}\cdot\vec{\sigma}]= \braket{\vec{v}}{\vec{w}}_{\bb{R}^3}.
\end{align}
Hence, $R(\Pi_U)\in SO(3)$, the set of isometries on $\bb{R}^3$. The other way around, for any $R\in SO(3)$ we can define a map $\Pi(R)=S\circ R\circ S^{-1}$ such that for any $\vec{v}\cdot\vec{\sigma}\in\bb{H}_3$ it holds 
\begin{align}
	\Pi(R)(\vec{v}\cdot\vec{\sigma})=(R\vec{v})\cdot\vec{\sigma}.
\end{align}
Clearly, as $\mcal{G}^+$ is the set of basis transformations in $\bb{H}_3$ we can translate the basis transformation in $\bb{R}^3$ induced by $R$ to a basis transformation $U(\cdot)U\degga$ in $\bb{H}_3$ with $U\in SU(2)$ such that $\Pi(R)=\Pi_U$. 
Hence, the isometry $S$ induces a group isomorphism which we also denote by $S:\mcal{G}^+\rightarrow SO(3),\;\Pi_U\mapsto R(\Pi_U)$, where the homomorphic structure follows from the properties of the scalar product. Finally, we can define  $S^\star=S\circ\iota\circ\rho:SU(2)/\bb{Z}_2\rightarrow SO(3)$ which defines the claimed group isomorphism.
\end{proof}

\subsection{SPT phases of $\mathcal{G}_{\rm Tot}$}\label{sptPhasesG}
We now argue that the 
SPT phases with respect to  $\mcal{G}$ are embedded in the SPT phases of $\mcal{G}_{tot}$, i.e.,~we explain how the additional discrete symmetries of the system affect the cohomology. 
To do so, we employ K\"unneth's theorem for cohomology which states that for two groups $G_1$ and $G_2$, and corresponding $G$-module $M$ \cite{75485}
\begin{widetext}
\begin{align}\label{kunneth}
H^n(G_1\times G_2,M)\simeq 
\left[\bigoplus_{i+j=n}H^i(G_1,M)\otimes H^j(G_2,M)\right]\oplus \
\left[\bigoplus_{p+q=n+1}Tor_1^\bb{Z}(H^p(G_1,M),H^q(G_2,M))\right].
\end{align}
\end{widetext}
Eq.~\eqref{kunneth} holds true even for non-trivial actions of the $G_i$ on $M$ given that one of the $G$-modules is $\mathbb{Z}$-free. 
In our case  we use $G_1\times G_2=\mathcal{G}_{tot}$ with $G_1=SO(3)$, $G_2$ being the remaining discrete and finite symmetries and $M=U(1)$. The action on $M$ is trivial for all symmetries but TR. 

Note that the Tor functor $Tor_m^R(A,B)$ is trivial given $A$ is a free $R$ module and $m\geq1$. Moreover, $Tor_m^R$ is symmetric given $R$ is abelian. Hence, as $G_2=\bb{Z}_2\times\bb{Z}_2\times\bb{Z}_2^{TR}$ and $R=\bb{Z}$,
the Tor functor in Eq.~\eqref{kunneth} is trivial in our case\cite{ChenGuWen2011, 75485}. 
Similarly, we find that $H^1(SO(3),U(1))$ as well as $H^1(\bb{Z}_2\times\bb{Z}_2\times\bb{Z}_2^T,U(1))$ is trivial, such that the possible SPT phases in our system reduce to 
\begin{align}
	&H^2(\mcal{G}_{tot},U(1))=H^2(SO(3),U(1))\oplus \notag \\
	&H^2(\bb{Z}_2\times\bb{Z}_2\times\bb{Z}_2^T,U(1)).
\end{align}
In particular, as $H^2(SO(3),U(1))$ is a subgroup of this, the group element corresponding to the Haldane phase is contained in $H^2(\mcal{G}_{tot}, U(1))$ which corresponds to the possible non-trivial SPT phase in our system.

\subsection{Detection of SPT phases}\label{derivationOrderParameter}
In this section we are going to derive an order parameter for the $SO(3)$ symmetry characterizing the Haldane phase that 
emerges due to the blocking of opposing lattice sites. 
We therefore follow the construction of the $SO(3)$ order parameter for the standard Haldane phase in the spin 1 chain as described in \cite{EntanglementSpectrumTopologicalPhase, PollmannTurner}. As shown in 
Ref.~\cite{SymmetriesMPS}, 
given a uMPS parametrized by the 3-tensor $M$ and a local (anti-) unitary symmetry operation $\mcal{G}$ parametrized by the matrix $g$ on the physical level (and complex conjugation of $M$ in case of anti-unitaries), then $\mcal{G}$ acts on the uMPS as explained in Fig.~\ref{fig:SymMPSVirtRep}. 

\begin{figure}
 \includegraphics{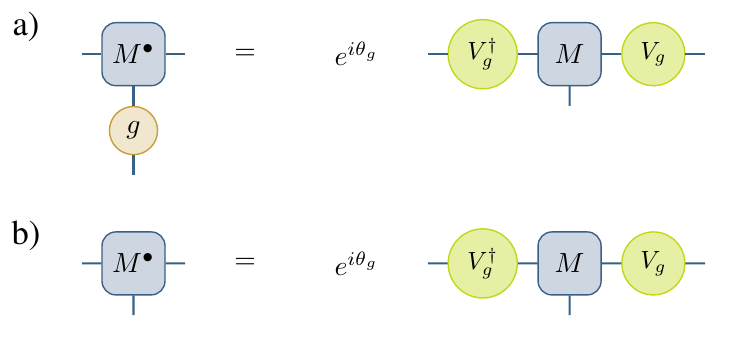}
 \caption{Illustration of the action of the symmetry group on a uMPS using the usual diagrammatic notation of tensor networks. 
		Given a MPS characterized by a 3-tensor $M$, we obtain the virtual representation $V_g$ as displayed in the upper panel a) from the action of the symmetry operation $g$ on the MPS tensor. If the state is symmetric under the action of $g$ the relation simplifies as shown in the lower panel b).}
 \label{fig:SymMPSVirtRep}
\end{figure}

Let us now consider the action of $\mathcal{G} \simeq SO(3)$ looking at the two elements 
\begin{align}
	\mathcal{R}_j = \exp(i\pi \sigma_j/2) \otimes \exp(i\pi \sigma_j/2)\in\mathcal{G}
	\label{math:physicalSymmetry}
\end{align}
for $j=x,z$. 
By the isomorphism of $\mcal{G}$ and $SO(3)$ explained in Sec.~\ref{sec:IsomorphismOfGAndSO3}, we can interpret those operators as $\pi$ rotations around the $j$ axes in $\bb{R}^3$.
Using the relation displayed in Fig.~\ref{fig:SymMPSVirtRep} a) we obtain for each $j$ the respective virtual representation $V_j$. Iterating the relation in Fig.~\ref{fig:SymMPSVirtRep} b) twice and using $[V_j,\Lambda]=0$ as well as the defining formulas for the left canonical gauge
\begin{align}
	\sum_j M^\dagger_j \id M_j = \id\;;\quad\sum_j M_j \Lambda M^\dagger_j = \Lambda
	\label{mpsGauge}
\end{align}
we obtain $V^2_j=e^{i\theta_j}\id$. 
Hence, using the $U(1)$ gauge freedom of the $V_j$'s we can substitute $V_j\mapsto e^{-i\theta_j/2}V_i$ such that we assume $V^2_j=\id$ in the following. 
Consecutively applying $\mathcal{R}_x$ and $\mathcal{R}_z$ and iterating the relation in Fig.~\ref{fig:SymMPSVirtRep} b) twice and using Eq.~\eqref{mpsGauge}, however, we obtain $V_xV_zV_xV_z=e^{i\theta_{x,z}}\id$. 
As the gauge of the $V_j$ is already fixed the phase $\theta_{x,z}$ can not be absorbed in the $V_j$. 
Moreover, using the unitarity of the $V_j$'s we can rewrite this in a gauge invariant form
\begin{align*}
	e^{i\theta_{x,z}}\id&=V_xV_zV_xV_z\\
	&=V_xV_ze^{i\theta_x/2}e^{i\theta_z/2}e^{-i\theta_x/2}e^{-i\theta_z/2}V_xV_z\\
	&=\tilde{V}_x\tilde{V}_z\tilde{V}^\dagger_x\tilde{V}^\dagger_z
\end{align*} 
where the $\tilde{V}_j$'s are the virtual representation in an arbitrary gauge, such that we obtain
\begin{align}
	\tilde{V}_x\tilde{V}_z=e^{i\theta_{x,z}}\tilde{V}_z\tilde{V}_x.
\end{align}
Obviously, the phase $\theta_{x,z}$ is connected to the cocycle evaluated at $\omega(\mathcal{R}_x,\mathcal{R}_z)$\cite{EntanglementSpectrumTopologicalPhase}. Hence, any $\theta_{x,z}\neq2\pi n$ corresponds to a non-trivial cocycle and it turns out that for $SO(3)$ it holds that $e^{i\theta_{x,z}}=1,-1$ where $-1$ corresponds to the topological non-trivial Haldane phase. Keeping this in mind it is 
straightforward to define the $SO(3)$ order parameter as
\begin{align}
	\mathcal{O}_{SO(3)}=\frac{1}{D}Tr\left[\tilde{V}_x\tilde{V}_z\tilde{V}^\dagger_x\tilde{V}^\dagger_z\right]
	\label{orderParameter}
\end{align}
for symmetric MPS and $0$ else, where we use the gauge free representation $\tilde{V}_j$.

In order to make this more precise, from K\"unneth's theorem we find that
\begin{align}\label{kunneth2}
	&H^2(\mathcal{G}_{tot},U(1))\simeq \notag\\
	&H^2(SO(3),U(1)) 
	\oplus H^2(\bb{Z}_2^L\times\bb{Z}_2^T\times\bb{Z}_2^F,U(1)),
\end{align}
as explained in \ref{sptPhasesG}. Hence, given we find a non trivial phase $\mathcal{O}_{SO(3)}=-1$ for some $g,h\in\mathcal{G}_{\rm Tot}$ of the form $g=g_{SO(3)}e_{\bb{Z}_2^L\times\bb{Z}_2^T\times\bb{Z}_2^F}$ and $h=h_{SO(3)}e_{\bb{Z}_2^L\times\bb{Z}_2^T\times\bb{Z}_2^F}$, we can deduce from this that the corresponding factor $\theta_{gh}$ is corresponding to a non trivial element from $H^2(\mathcal{G}_{\rm Tot},U(1))$ that is from the $H^2(SO(3),U(1))$ sector in ($\ref{kunneth2}$). In other words we deduce from this the fact that we are in a topologically non trivial SPT phase with respect to $SO(3)$, the Haldane phase. 
Now, this order parameter allows us to compare the SPT phase of the blocked model with the phase of an effective $SO(3)$ symmetric spin-$1$ model in order to explore the validity of the effective model as we will illustrate in more details in the next part on the example of the bi-layered delta chain.

\section{Case study: a bi-layered $\Delta$-chain}
In the focus of attention in this work, as a proxy for similar microscopic models allowing for an effective description,
 are \emph{bi-layered $\Delta$-chains} (BDC) with \emph{anti-ferromagnetic Heisenberg interaction} 
within the layer and \emph{ferromagnetic Heisenberg interaction} between the layers (see  Fig.~\ref{lattice}). The corresponding model Hamiltonian is given by 
\begin{align}
	H=J_{AF}H_{AF}+J_F H_F
\end{align} 
where the ferromagnetic and anti-ferromagnetic Heisenberg terms $H_F$ and $H_{AF}$ are given by 
\begin{align}
&	H_{AF}=\sum_{\langle i,j \rangle \in E_{AF}}\vec{s}_i\cdot\vec{s}_j\label{eq:HAFDef},\\
&	H_{F}=-\sum_{\langle i,j \rangle \in E_{F}}\vec{s}_i\cdot\vec{s}_j.
\end{align}
Here $E_{AF}$ denotes the anti-ferromagnetic edge set (blue bonds in Fig.~\ref{lattice}) and $E_F$ denotes the ferromagnetic edge set (red bonds in Fig.~\ref{lattice}). Moreover, throughout the paper we use the convention $J_{AF}, J_F \geq 0$.

\begin{figure}
	\includegraphics{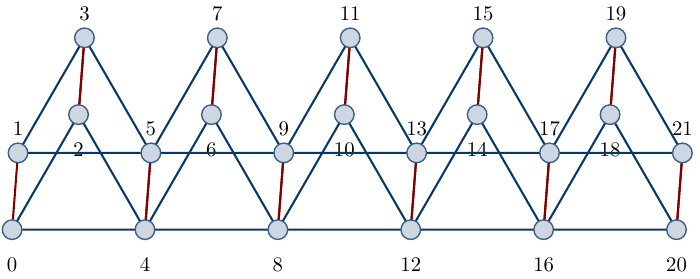}
	\caption{Illustration of the double layered $\Delta$-chain lattice with 22 sites and open boundary conditions. On the blue bonds (within the $\Delta$-chain layers), collected in the edge set $E_{AF}$, we assume an anti-ferromagnetic Heisenberg interaction where on the red bonds between opposing vertices, which are collected in the edge set $E_F$, a ferromagnetic Heisenberg interaction is introduced. In case of periodic boundary conditions, which we focus on in this work, the sites 20 and 21 are identified with the sites 0 and 1 respectively.}
	\label{lattice}
\end{figure}

The ratio $J_F/J_{AF}$ determines the physics of the model. 
Therefore, in view of performing a numerical investigation of the system we chose to parametrize the couplings using a compactly supported parameter $\theta\in[0,\pi/2]$ by defining
\begin{align}
	H(\theta)=\cos(\theta)H_{AF}+\sin(\theta)H_F.
\end{align}
Here, $\theta$ interpolates between weak ferromagntic inter-layer couplings ($\theta\approx 0$) and very strong ferromagnetic interactions ($\theta\approx \pi/2$).
As explained in the introduction, we are going to draw the connection between the original double layered spin-$1/2$ model and the effective single-layered spin-$1$ model on the level of phases. 
To be precise, the effective model here is defined as a single-layered delta chain with spin-$1$ particles on each site and anti-ferromagnetic coupling just as in Eq.~\eqref{eq:HAFDef}.
In order to establish the comparison of the effective and the original model, it is crucial to introduce the blocking in the original spin-$1/2$ model of the ferromagnetically coupled edges into a single unit cell. 
This leads to a new Hamiltonian describing a single-layered $\Delta$-chain of generalized quantum particles supported on a local Hilbert space $\cc^2\otimes\cc^2$. 
For the sake of distinction we will refer to the blocked Hamiltonian as $H_b(\theta)$. 
The main reason for the emergence of the Haldane phase in a system basically described by a spin-$1/2$ ladder can be traced to the annihilation of the phases in the tensor product $g\otimes g$.

In order to obtain an intuition for the BDC we are first going to analytically investigate its strong and weak coupling limits in the following three sections. 
Subsequently we employ numerical MPS based methods in order to investigate the order parameter of the blocked model probing the validity of the effective spin-$1$ model explained above.

\subsection{single-layered spin-$1/2$ and uncoupled $\Delta$-chains}\label{section:singleLayer}
To start with, we consider an isolated spin-$1/2$ $\Delta$-chain with anti-ferromagnetic Heisenberg interaction. 
We can write the Hamiltonian of the system as $H_{\rm single}=\sum_\Delta h_\Delta$, where we sum over each triangle in the chain, and $h_\Delta$ refers to the local Hamiltonian on the respective triangle $h_\Delta=\vec{s}_1\cdot \vec{s}_2+\vec{s}_2\cdot \vec{s}_3+\vec{s}_3\cdot \vec{s}_1 = \frac{1}{2} \vec{s}_\Delta^2+const.$ 
When coupling two spin-$1/2$ particles the respective Hilbert space decomposes as  $\frac{1}{2}\otimes\frac{1}{2}=0\oplus1$. 
Doing the same with three spins we end up with $\frac{1}{2}\otimes\frac{1}{2}\otimes\frac{1}{2}=\frac{1}{2}\oplus\frac{1}{2}\oplus\frac{3}{2}$. 
We can easily construct a (non-orthogonal) basis for the ground state space of a single triangle as follows. 
Any ground state has a total spin of $s_{\rm tot}=\frac{1}{2}$. 
So we take two of the three spins and couple them in a singlet.
The third spin can be chosen arbitrarily, such that the three spins form a $s_{\rm tot}=\frac{1}{2}$ state. In particular, we find four independent states by putting the singlet either on the edge $\av{1,2}$ or $\av{2,3}$ and putting the third spin either in the state vector $\ket{0}$ or $\ket{1}$. In particular, in what follows we will denote these states on a single triangle as 
\begin{align}
	&\ket{\nearrow}\otimes\ket{s}=\frac{1}{\sqrt{2}}\left( \ket{0,1} - \ket{1,0} \right)\otimes\ket{s},\\ 
	&\ket{s}\otimes\ket{\nwarrow}=\frac{1}{\sqrt{2}}\ket{s}\otimes\left( \ket{0,1} - \ket{1,0} \right)
\end{align}
where the arrow indicates a singlet, which we will also refer to as dimer, on the up, respectively downwards facing edge and $s=0,1$. It is easy to compute
\begin{align}\label{math:singletOverlap}
	\abs{(\bra{s}\otimes\bra{\nwarrow})(\ket{\nearrow}\otimes\ket{\tilde{s}})} \leq \frac{1}{\sqrt{2}}.
\end{align}

Using this observation we can construct the ground states of the full (periodic) chain of $N$ sites as follows. Either we place a singlet on any of the up  edges, or we put it on any of the down edges. 
That these states correspond to ground states of the system is clear from the observation that they are locally the ground states of each of the $h_\Delta$ and therefore the minimal possible energy expectation value and it is also clear that those are the only states corresponding to this energy.
Hence, periodic single-layered $\Delta$-chain has a two fold degenerate ground state space spanned by the dimer states formed by singlets sitting on all up or down edges. We denote these states by
\begin{align}
	&\ket{\nearrow}_N=\ket{\nearrow}^{\otimes N/2}\label{math:ketUp}\\
	&\ket{\nwarrow}_N=\ket{\nwarrow}^{\otimes N/2}\label{math:ketDown}
\end{align}
where we will drop the index $_N$ if it is clear from the context that we talk about the many particle wave function. It is worth mentioning that the state $\ket{\nwarrow}$ in the definitions \eqref{math:ketUp} and \eqref{math:ketDown} is implicitly assumed to be placed on edges $\langle 2n-1, 2n\rangle$ facing upwards whilst the state $\ket{\nearrow}$ is assumed to be placed on edges $\langle 2n, 2n+1\rangle$ facing downwards. This implies that the tensor products in the definitions \eqref{math:ketUp} and \eqref{math:ketDown} act between different local Hilbert spaces. 
Moreover, it is easy to see that the ground state degeneracy for an open $\Delta$-chain would grow proportional in system size as for any site $2i+1$ the state $\ket{\nearrow}^{\otimes i}\otimes\ket{s}\otimes\ket{\nwarrow}^{\otimes N/2 - i}$ is a ground state. 

Having identified the ground state of a single chain we can easily define the ground state space of the uncoupled bi-layered $\Delta$-chain. As the Hamiltonian of the uncoupled system is simply the sum of the individual Hamiltonians of each layer $H_{AF}=H_{\rm single}\otimes \id + \id\otimes H_{\rm single}$, the ground state space is spanned by the tensor products of the ground states of the individual layers. The two ground states of a single layer give therefore rise to four states spanning the ground state space of the bi-layered system.

Using the notation introduced above we can write a basis spanning the four dimensional ground state space of the uncoupled BDC as
\begin{align}
&\ket{1} := \ket{ \nearrow} \otimes \ket{ \nearrow}  \label{unperturbedGroundState1} ,
&\ket{2} := \ket{ \nearrow} \otimes \ket{ \nwarrow},\\
&\ket{3} := \ket{ \nwarrow}\otimes \ket{ \nwarrow}  ,
&\ket{4} := \ket{ \nwarrow}\otimes \ket{ \nearrow} \label{unperturbedGroundState4}.
\end{align}
It is worth mentioning that the scalar product between any of those states vanishes exponentially in the thermodynamic limit as
\begin{align}
	\abs{{}_{N}{\braket{i}{j}}{_N}}\leq \frac{1}{\sqrt{2}^{N/4}}\;,\quad i\neq j
\end{align}
for $i,j=1,\dots,4$. Note that the state vectors $\ket{1}$ and $\ket{3}$ are of product form with respect to the unit cells introduced from the blocking of opposing spins and therefore correspond to a trivial SPT phase, where the other two states correspond to the non-trivial phase. 

\subsection{Strong coupling limit}
Let us now characterize the ground states in the extremal coupling regimes. 
To start with we analyzed the strong coupling limit. 
Rewriting the Hamiltonian of the system as $H= J_F(({J_{AF}}/{J_F}) H_{AF} + H_F)$ we find $H_{AF}$ to be infinitely suppressed in the infinite coupling limit $J_F\rightarrow\infty$ with $J_{AF}=const.$ Investigating  $H_\infty=({J_{AF}}/{J_F}) H_{AF} + H_F$ for $J_{AF}=1$ and large $J_F$ and using the local Hilbert space decomposition $\frac{1}{2}\otimes\frac{1}{2}=0\oplus1$ for two opposing vertices, we find that $H_\infty$ equals the effective model, the single-layered spin 1 $\Delta$-chain in the spin 1 subspace of the opposing vertices, plus a weak perturbation connecting the spin 0 and spin 1 spaces which is suppressed by $J_F$. Hence, to zeroth order perturbation theory the ground state in the strong coupling limit is the ground state of the effective spin 1 model mapped onto the double layered $\Delta$-chain. Using the TDVP for uMPS\cite{2011PhRvL.107g0601H} we numerically computed the ground space for the spin 1 system. 
Unsurprisingly, as the spin 1 $\Delta$-chain is an one dimensional anti-ferromagnetic spin-$1$ Heisenberg model, its ground state is found to be in the same phase 
as the anti-ferromagnetic spin-$1$ Heisenberg model on a linear chain, the topologically 
non-trivial Haldane phase.
From this consideration it is however unclear to what extend this result caries over to other coupling strengths as the ground state space of the uncoupled chain is four-fold degenerate and therefore in a different phase. In what follows we will investigate the weak coupling limit in order to show that a weak ferromagnetic interaction favours the space spanned by the topologically trivial states and leads to a symmetry broken phase. Hence, for some finite coupling strength a transition from symmetry broken to the SPT phase has to occur in the BDC.

\subsection{Weak coupling limit}\label{sec:Weak coupling limit}
Let us now investigate the weak coupling limit.
We denote the ground state space spanned by the state in Eq.~\eqref{unperturbedGroundState1}, \eqref{unperturbedGroundState4} of $H_{AF}$ the uncoupled periodic BDC with $N$ sites by $\mcal{H}_0$ and consider the perturbation $V = H_F$ to the uncoupled BDC.
As we find $\bra{\psi}V_e\ket{\psi}=0$ for every ferromagnetic edge $e$ and all $\ket{\psi}\in\mcal{H}_{0}$ we need to find the minimizer of the second order $\bra{\psi}V(H_{AF}-E_0)^{-1}P_1\ket{\psi}$ with $\ket{\psi}\in\mcal{H}_0$. 
Here $P_1$ is the projection onto $\mcal{H}_0^\perp$ the orthogonal complement of the ground state space of $H_{AF}$ and $E_0$ is its ground state energy. 
Therefore we define the matrix $(\Delta_{i,j})$ in the thermodynamic limit as $\Delta_{i,j}=\lim_{N\rightarrow\infty}\Delta_{i,j}(N)$ with
\begin{align}
	\Delta_{i,j}(N) = \matrixelm{i}{V(H_{AF}-E_0)^{-1}P_1V}{j}, \label{energyCorrectionMatrix}
\end{align}
where $i,j=1,2,3,4$ refer to the ground states of the unperturbed Hamiltonian as defined in Eq.~\eqref{unperturbedGroundState1}, \eqref{unperturbedGroundState4}. 
Our aim is now to give an estimate of the entries $\Delta_{i,j}$.

Writing $P_M$ for the projection onto the maximal eigenspace of $H_{AF}$ and $P_M^\perp$ for the projection onto the orthogonal complement and using $P_MP_1=P_M$ we can rewrite \eqref{energyCorrectionMatrix} as follows
\begin{align}
	\Delta_{i,j}(N) = & -\frac{1}{E_0}\matrixelm{i}{V\left(\id -\frac{H_{AF}}{E_0}\right)^{-1}P_1V}{j} \notag\\
	 = & -\frac{1}{E_0}\matrixelm{i}{V\left(\id -\frac{H_{AF}}{E_0}\right)^{-1}P_MV}{j} \notag\\ 
	&- \frac{1}{E_0}\matrixelm{i}{V\left(\id -\frac{H_{AF}}{E_0}\right)^{-1}P_M^\perp P_1V}{j}.
\end{align}
Next, we observe that the energy of $H_{AF}$ locally on each of the $\frac{N}{2}$ triangles is in $[-\frac{3}{4},\frac{3}{4}]$ with 
\begin{equation}
E_{\rm max}^{\rm local}=-E_{\rm min}^{\rm local}=\frac{3}{4}
\end{equation}
and hence $H_{AF}+\abs{E_0}\geq0$ (and, as is easy to show, $\abs{E_0}=E_{\rm max}$). Therefore, we can expand the last term into a Neumann series as 
\begin{equation}
\matrixelm{\phi}{\frac{H_{AF}}{E_0}}{\psi}/\braket{\phi}{\psi}<1
\end{equation}
 for all $\psi\in{\rm Img}(P_M^\perp P_1)$. We find
\begin{align}\label{math:perturbationMatrix}
	\Delta_{i,j}(N) = & -\frac{1}{E_0}\matrixelm{i}{V\left(\id -\frac{H_{AF}}{E_0}\right)^{-1}P_MV}{j} \notag\\ 
	&  - \frac{1}{E_0}\sum_{k}\matrixelm{i}{V \left(\frac{H_{AF}}{E_0}\right)^k P_M^\perp P_1V}{j}.
\end{align}
It is easy to see that
\begin{align}
 &\frac{1}{\abs{E_0}}\sum_k\abs{ \matrixelm{i}{V\left(\frac{H_{AF}}{E_0}\right)^kP_M^\perp P_1V}{j} }\notag\\
	&\qquad\leq \frac{1}{\abs{E_0}}\sum_k \norm{V\ket{i}}^2\norm{\left.\frac{H_{AF}}{E_0}\right|_{{\rm Img}(P_M^\perp P_1)}}^k
\end{align}
such that, due to the projectors, by construction for every $N$ the sum over $k$ in Eq.~\eqref{math:perturbationMatrix} is absolutely converging.

Next we estimate the first summand in  Eq.~\eqref{math:perturbationMatrix} as well as each of the terms in the von Neumann series in the large $N$ limit. Therefore, we make use of the fact that the overlap of the excitations and the ground state space and the maximal energy eigenspace is small. In particular, 
\begin{align}
	\abs{\matrixelm{i}{V}{j}} \leq \sum_a\abs{\matrixelm{i}{V_a}{j}} \leq \frac{N}{\sqrt{2}^{N/4-1}} (1-\delta_{i,j})
\end{align}
as for $i=j$ for each of the summands the vector contains a triplet being orthogonal to the singlet in the dual vector, 
whilst for $i\neq j$ there are at least $N/4-1$ triangles being covered with different singlet configurations (cf.~Eq.~\eqref{math:singletOverlap}). Similarly, for any $\ket{M}\in {\rm Img}(P_M)$ it holds
\begin{align}
	\matrixelm{M}{V}{i} = \sum_a \matrixelm{M}{V_a}{i} = 0
\end{align}
as each triangle in $\ket{M}$ is covered by a spin $\frac{3}{2}$ configuration being orthogonal to the singlet configurations in $V_a\ket{i}$. Therefore, using the continuity of the scalar product, we can drop the projections $P_1$ and $P_M^\perp$  in the thermodynamic limit and drop the first summand in  Eq.~\eqref{math:perturbationMatrix}. We obtain
\begin{align}
	\Delta_{i,j} &= -\sum_k\lim_{N\rightarrow\infty}\frac{1}{E_0}\matrixelm{i}{V\left(\frac{H_{AF}}{E_0}\right)^k V}{j} \notag\\
	&= -\sum_k\lim_{N\rightarrow\infty}\frac{1}{E_0}\sum_{a,b}\matrixelm{i}{V_a\left(\frac{\sum_\Delta h_\Delta}{E_0}\right)^k V_b}{j}\label{math:offDiagonalEst1_1}
\end{align}

Let us now estimate the off diagonal elements $i\neq j$. Assume therefore that $\ket{i}=\ket{3}$ and $\ket{j}=\ket{1}$, where the other cases work in the same fashion and only the exponent in the final estimation might change by a factor of two. 
Expanding the power of sums over triangles we can bound
\begin{align}
 &\abs{\matrixelm{3}{V_a\left(\frac{\sum_\Delta h_\Delta}{E_0}\right)^k V_b}{1}}\notag\\
 &\qquad\leq\sum_{i\in \left[\frac{N}{2}\right]^k} \abs{\matrixelm{3}{V_a\frac{h_{\Delta_{i_1}}\dots h_{\Delta_{i_k}}}{E_0^k}  V_b}{1} }.\label{math:offDiagonalEst1_2}
\end{align}
Moreover, it is easy to compute
\begin{align}
	&V_a\ket{\nearrow}_N\otimes \ket{\nearrow}_N =\ket{\nearrow}_{N-1}\otimes \ket{\nearrow}_{N-1}\otimes \notag\\
	&\otimes\left(\sum_{i,j=-1}^1\alpha_{i,j}\ket{_i\diagup}\otimes\ket{_j\diagup}\right)_a
\end{align}
where $\ket{_i\diagup}$ denotes a triplet with $m_z=i$ sitting on the upwards facing edge of a triangle (similarly $\ket{\diagdown_i}$ for the downwards facing edge), and $\alpha_{i,j}=\frac{1}{4}\delta_{i,-j}(-1)^i$. Also we implicitly assume the tensor products to be ordered in such a way that the last factor on the right hand side corresponds to the pair of triangles for which the singlets (before the action of $V_a$) are intersecting the ferromagnetic edge $a$, indicated by the index $a$. Then we can rewrite the summands in \eqref{math:offDiagonalEst1_2} as 
\begin{widetext}
\begin{align}
	&\matrixelm{3}{V_a\frac{h_{\Delta_{i_1}}\cdots h_{\Delta_{i_k}}}{E_0^k}  V_b}{1} =\notag\\
	&= {}_{N-1}{\bra{\nwarrow}}\otimes {}_{N-1}{\bra{\nwarrow}}\otimes\left(\sum_{i,j=-1}^1\alpha_{i,j}\bra{\diagdown_i}\otimes\bra{\diagdown_j}\right)_a \frac{h_{\Delta_{i_1}}\cdots h_{\Delta_{i_k}}}{E_0^k}\ket{\nearrow}_{N-1}\otimes \ket{\nearrow}_{N-1}\otimes\left(\sum_{i,j=-1}^1\alpha_{i,j}\ket{_i\diagup}\otimes\ket{_j\diagup}\right)_b . \label{math:preBound}
\end{align}
\end{widetext}
We know that each $h_\Delta\ket{\nearrow}\otimes\ket{s}=E_{\rm min}^{\rm local}\ket{\nearrow}\otimes\ket{s}$ independent of $s=0,1$ and similarly for the other singlet configuration. $h_\Delta$ can act non-trivially only on triangles covered by a triplet and can propagate the excitation only to one of the neighbouring triangles (cf. Appendix \ref{append:localDynamics}). Henceforth, $h_{\Delta_{i_1}}\dots h_{\Delta_{i_k}}$ can at most create a set of $k$ triangles of the $N/2$ triangles that may not be in the singlet configuration. 
Let now $S$ be the set containing the $4$ triangles excited by the local perturbations plus the at most $k$ triangles which are acted on non-trivially by the $h_\Delta$. 
Then we can bound Eq.~\eqref{math:preBound} using $\norm{h_\Delta}=\frac{3}{4}<1$, the Cauchy-Schwarz inequality and $E_0=\frac{N}{2} \norm{h_\Delta}$
\begin{widetext}
\begin{align}
	&\abs{\matrixelm{3}{V_a\frac{h_{\Delta_{i_1}}\cdots h_{\Delta_{i_k}}}{E_0^k}  V_b}{1}} \leq  \left| {}_{N-k-4}{\braket{3}{1}}{_{N-k-4}} \; {}_{k_1}{\bra{\nwarrow}}\otimes{}_{k_2}{\bra{\nwarrow}}\otimes\left(\sum_{i,j=-1}^1\alpha_{i,j}\bra{\diagdown_i}\otimes\bra{\diagdown_j}\right)_a  \right.\notag\\
	&\quad\left.\times\frac{h_{\Delta_{i_1}}\cdots h_{\Delta_{i_k}}}{E_0^k} \ket{\nearrow}_{k_1}\otimes \ket{\nearrow}_{k_2}\otimes\left(\sum_{i,j=-1}^1\alpha_{i,j}\ket{_i\diagup}\otimes\ket{_j\diagup}\right)_b\right| \leq \frac{2^k}{N^k\sqrt{2}^{N/2-k-4}}, \label{math:offDiagonalEst2}
\end{align}
\end{widetext}
where the states $\ket{j}_{N-k-2}$ are meant to be supported on the complement of $S$ only,
while the triangles in $S$ that are not excited by the excitation are labeled by $k_1$ and $k_2$ fulfilling $k_1+k_2=k+2$. 
If we had chosen a different combination of states $\ket{i}$ and $\ket{j}$ in the beginning, the same argumentation would hold just that for a combinations such as $\ket{1}$ and $\ket{2}$, where the states coincide in one of the layers, the exponent in the final estimate in Eq.~\eqref{math:offDiagonalEst2} changes from $N/2-k$ to $N/4-k$, as then only half the triangles are populated by a different covering.
Finally, combining this bound with \eqref{math:offDiagonalEst1_1} and \eqref{math:offDiagonalEst1_2} we can conclude that
\begin{align}
	|\Delta_{i,j}| \leq \sum_k \lim_{N\rightarrow\infty} \frac{N^2}{\sqrt{2}^{N/4-k-4}} = 0
\end{align}
 for $i\neq j$.

Let us now investigate the diagonal elements $\Delta_{i,i}$. 
For those it is crucial to investigate the action of $H_{AF}$ onto the states $V_a\ket{i}$. 
In Appendix \ref{append:localDynamics} we investigate the action of the local terms $h_\Delta$ on locally excited states of the form $\ket{{}_i{\diagup}}\otimes\ket{\nearrow}_{N-1}$, the other configuration follows directly from inversion symmetry. 
We find that the excitation may be spread only to the neighbouring triangle to the right, while the triangle originally occupied by the excitation will always remain occupied by some triplet excitation.
In particular, no combination of local terms can transform a triplet on the original triangle into a singlet state. 
This can be used straightforward in order to calculate the corresponding action of the local terms $h_\Delta$ in the BDC on the states $V_a\ket{i}$. 
It follows directly that $V_a\ket{i}$ can recombine only with itself, respectively that
\begin{align}
	\bra{i}V_aV_b\ket{i}=0\quad\Rightarrow\quad\bra{i}V_aH_{AF}^kV_b\ket{i}=0
\end{align}  
Now it is easy to verify that for the parallel configurations $V_a\ket{i}=V_b\ket{i}$ for $a=b$ as well as if $b$ equals the ferromagnetic edge next to $a$ intersecting the same pair of singlets. For the alternating configuration, however, $V_a\ket{i}\perp V_b\ket{i}$ for every $b\neq a$. Hence, for every $k$ exactly twice as many terms survive the sum over pairs of ferromagnetic edges for the parallel configurations as opposed to the alternating configuration.

Next we make use of the fact that the locally excited states $V_a\ket{i}$ are exponentially located in the low energy spectrum of $H_{AF}$. In particular, we can decompose
\begin{align}
	\bra{i}V_a \left(-H_{AF}\right)^k V_a\ket{i} = V_>^{(k)} - V_\leq^{(k)}
\end{align}
into positive and semi-negative parts. 
Note that for $k$ even, the negative part vanishes trivially as $(-H_{AF})^k$ is a positive operator then.
Assuming $k$ odd and using (in the last estimate) Theorem 2.1 from \citep{Arad2016}, $\norm{V_a}=\frac{1}{4}$ and $\norm{H_{AF}}=\frac{3}{8}N$ we find
\begin{align}
	\abs{V_\leq^{(k)}} &= \abs{\bra{i}V_aP_\geq(-H)^kP_\geq V_a\ket{i}} \leq \norm{H}^k \norm{P_\geq V_a\ket{i}}^2\notag\\
	&=\norm{H}^k \norm{P_\geq V_aP_0\ket{i}}^2 \leq \norm{H}^k \norm{P_\geq V_aP_0}^2\notag\\
	&\leq  c(3cN/2)^ke^{-\lambda(3cN/2-2R)} \propto N^ke^{-N}
\end{align}
with 
\begin{align}
c=\frac{1}{4},\, R=12c,\,\lambda=\frac{1}{2gk^\prime},\,g=6c
\end{align}
and $k^\prime=3$\footnote{Here, we have used the notation from Ref.~\citep{Arad2016} with $R=\sum_{\Delta\in E_V}\norm{h_\Delta}$ where $E_V$ is defined by $[H_{AF},V_a]=\sum_{\Delta\in E_V}[h_\Delta,V_a]$, $k^\prime=\max _\Delta\abs{{\rm supp}(h_\Delta)}$ and $g=\max_i\abs{\{h_\Delta:\;i\in{\rm supp}(h_\Delta)\}}$}, $P_0$ being the ground space projector and $P_\geq$ being the projector on the space corresponding to energies $E\geq0$. Hence, $V_\leq^{(k)}\rightarrow0$ in the thermodynamic limit, from which it follows that 
\begin{align}
	\biggl\langle i\biggl|V_a \biggl(\frac{H_{AF}}{E_0}\biggl)^k V_a \biggl| i\biggr\rangle\geq0.
\end{align}
We conclude that $\Delta_{i,j}\leq0$. Hence, the alternating configuration is suppressed in the weak coupling limit.

It is easy to see that the parallel configuration, as opposed to the alternating configuration, can be written as a tensor product state with respect to the blocked sites. Additionally, the representative states each break inversion symmetry. 
We therefore conclude this study with a two-fold degenerate symmetry broken ground state space in the weak coupling limit spanned by the two symmetry broken states $\ket{1}$ and $\ket{3}$.

\subsection{Matrix product state simulations}
In the proceeding sections we found that for very strong ferromagnetic couplings, the effective spin-$1$ model describes the physics of the original model well whereas at weak coupling strength the system is in a different SPT phase from the effective model.
Therefore, we simulated the BDC in order to verify our findings from the proceeding sections. Moreover, using the order parameter $\mcal{O}_{SO(3)}$ we identify the critical coupling strength from which on the effective model and blocked system are in the same SPT phase.
Using the time dependent variational principle (TDVP) for uMPS\cite{2011PhRvL.107g0601H} we compute the $\theta$-dependent ground state of the BDC. 
With this, we are then able to compute the projective symmetry representation\cite{PollmannTurner} characterizing the SPT phase.  
To do so, we use a modified version of the order parameter $\mathcal{O}_{SO(3)}$ derived in Sec.~\ref{derivationOrderParameter} as explained below. 
Using these numerical simulations we estimate the critical coupling strength for the transition to be $\theta_c\approx0.58$ corresponding to $J_F\approx0.66J_{AF}$. 
The corresponding data is shown in Fig.~\ref{phaseDiagram}.

\begin{figure}
\includegraphics{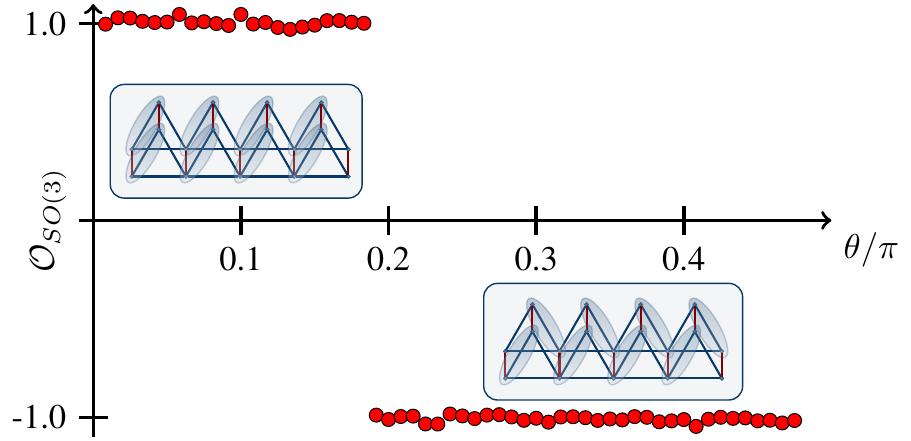}
\caption{Plot of the $\theta$ dependent order parameter evaluated for a uMPS approximation of the ground state of the blocked double layered $\Delta$-chain with bond dimension $D=64$ for $\theta$ linearly interpolating from $0$ to $\pi/2$ in $60$ steps. 
The data displays a clear change of the order parameter from $1$ (trivial phase) to $-1$ (Haldane phase) at around $\theta =0.18\pi$.
The dimer states in the insets show based on the dime-dimer correlation analysis presented in Fig.~\ref{fig:dimerDimer} the favored representative states for the respective phase corresponding to $\ket{1}$ for the trivial and $\ket{2}$ for the Haldane phase.}
\label{phaseDiagram}
\end{figure}

The BDC has a periodicity of four (respectively two per layer), such that  we use an enlarged unit cell containing four consecutive spin-$1/2$ vertices for a suitable uMPS description. 
Henceforth, we modified the order parameter derived in \ref{derivationOrderParameter} with respect to the physical symmetry used to derive the $V$ matrices exchanging \eqref{math:physicalSymmetry} by
\begin{align}
	\mcal{R}_j=\bigotimes_{k=1}^4 \exp(i\pi \sigma_j/2)
\end{align}
corresponding to the effective symmetry group
\begin{align}
	\tilde{\mcal{G}} &= \{g\otimes g\otimes g\otimes g:\; g\in SU(2)\}\nonumber\\
	&\simeq \{\tilde{g}\otimes \tilde{g}: \tilde{g}\in SO(3)\}
\end{align}
corresponding to a generalized local Hilbert space $(\cc^2)^{\otimes4}$. 
In the blocking of opposing vertices we essentially through the tensor product annihilated a negative sign in $SU(2)$ corresponding to  the quotient $SU(2)/\mathbb{Z}_2\simeq SO(3)$ as explained in Sec.~\ref{sec:IsomorphismOfGAndSO3}.
Hence, $\tilde{\mcal{G}}\simeq\mcal{G}$ as $\mcal{G}\simeq SO(3)$ contains no non-trivial normal subgroup that could be factored out through the additional tensor product. The SPT phases in this model are therefore the same as in the model with respect to $\mcal{G}$. 
However, in case of a non-trivial SPT phase further analysis is necessary in order to make sure that this phase is not depending on the blocking of four sites. 
This means, one has to make sure that the entanglement responsible for that phase is not only between the enlarged unit cells but also within the unit cell, such that a cut through it, mapping the enlarged unit cell back to the original blocked model, would become trivial. 

As a first sanity check we computed the uMPS representation of the ground state with respect to a shifted unit cell containing the last two vertices of the original unit cell and the first two vertices of the next unit cell (as explained in Fig.~\ref{fig:mpsShift}). 
Applying the modified $\mcal{O}_{SO(3)}$ on this representation we find the same behavior of the order parameter as shown in Fig.~\ref{phaseDiagram} corresponding to the original unit cell. This is strong evidence for the entanglement being non-trivial not only between the two consecutive four site unit cells but also within those unit cells. 

\begin{figure}
 \includegraphics{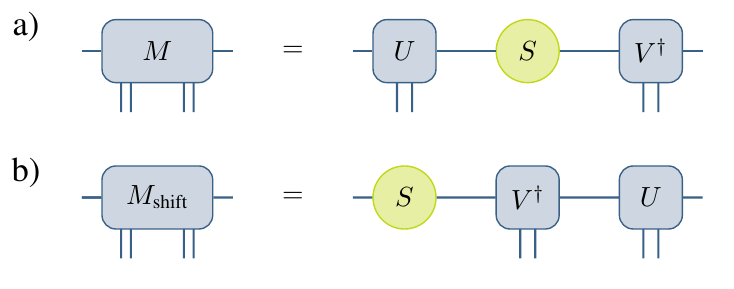}
 \caption{Illustration of the shift operation performed on the level of the uMPS tensor in order to test the order parameter for an alternative blocking. 
 Given the uMPS tensor $M$ we decompose it using a singular value decomposition as shown in panel a) into components that are associated to the different ferromagnetically coupled pairs of vertices. 
 We then define a new shifted uMPS tensor as shown in the lower panel b). Each of the matrices $U$ and $V$ has two physical legs highlighting the fact that we cut between two blocked unit cells containing two spin-$1/2$ spins each.}
 \label{fig:mpsShift}
\end{figure}

\begin{figure*}
	\includegraphics[width=\textwidth]{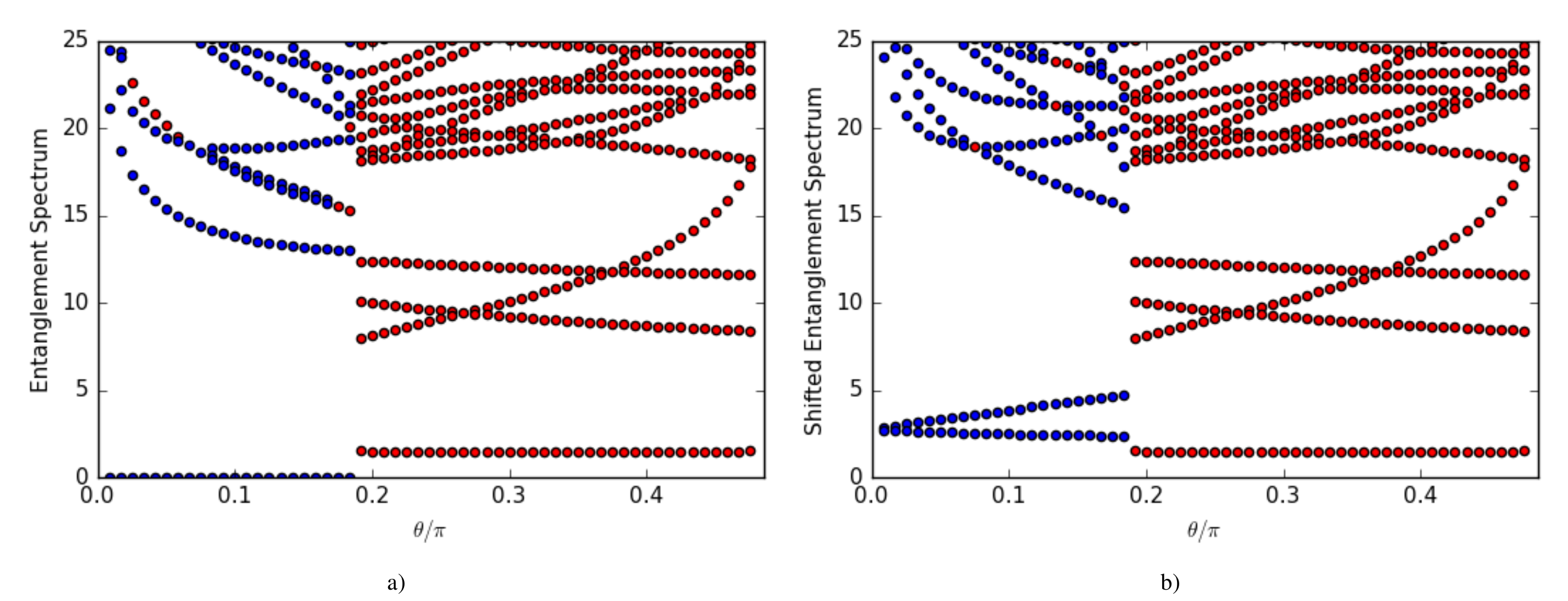}
	\caption{Plot of the $\theta$ dependent entanglement spectra over different cuts of the system based on the ground state approximation using a uMPS with bond dimension $D=64$. 
	In figure a) on the left we plot the entanglement spectrum computed as explained in Ref.~\cite{EntanglementSpectrumTopologicalPhase} of the state for a cut between the enlarged 4 site unit cells. Figure b) on the right displays the entanglement spectrum of the shifted uMPS corresponding to a bond dimension $D=256$ constructed as explained in Fig.~\ref{fig:mpsShift}, i.e.,~for a cut separating the two ferromagnetically coupled pairs of vertices inside one unit cell of the uMPS. In each plot a blue point indicates an oddly degenerate spectral value. Similarly, a red point marks an evenly degenerate spectral value. We counted two values to be degenerate if they deviate not more than $10^{-2}$ in relative error. 
	For both cuts one can clearly see the transition from a mostly oddly degenerate spectrum to a fully evenly degenerate spectrum around $\theta_c\approx0.18\pi$. This is consistent with the results for the order parameter presented in Fig.~\ref{phaseDiagram} and indicates a phase transition into an SPT phase.}
	\label{fig:ES}
\end{figure*}

To underline this we further analysed the entanglement  spectrum (ES) with respect to a cut between the original unit cells, as well as with respect to a cut through the four site unit cell between two neighbouring blocked sites. The results are shown in Fig.~\ref{fig:ES}. 
In accordance with the generalized order parameter, we find a transition from a dominantly oddly degenerate ES for $\theta<\theta_c$ to an evenly degenerate ES for $\theta>\theta_c$. 
For a cut through the enlarged unit cell we find an ES that is not degenerate at all levels (and in particular not evenly degenerate in the low energy levels) for $\theta<\theta_c$ whilst it is evenly degenerate for $\theta>\theta_c$ supporting evidence of a phase transition. 
We conclude from the ES, between and through the enlarged unit cells that the entanglement within the unit cell for $\theta>\theta_c$ is non-trivial. 
Combining this observation with \cite{EntanglementSpectrumTopologicalPhase} we find strong evidence for the state being in the Haldane phase with respect to the blocked two-site unit cell.

To align our numerics with the perturbative results, as well as to better understand the structure of these phases, we analyzed the dimer-dimer correlation function 
\begin{align}
	\mcal{D}(i,j) = \frac{1}{4} \langle \left( 2 - (s_i+s_{i+2})^2 \right) \left( 2 - (s_j+s_{j+2})^2  \right) \rangle
\end{align}
evaluated in the ground state, which yields the probability of having a dimer on the edge $\langle i,i+2 \rangle$ and a dimer on the edge $\langle j,j+2 \rangle$ (where even and odd sites correspond to different layers respectively as in Fig.~\ref{lattice}). 
The $\theta$ dependent results are shown in Fig.~\ref{fig:dimerDimer}. 
One can clearly see a transition from the system favoring a parallel dimer configuration to an alternating dimer configuration around $\theta_c$. 
Moreover we find the dimer-dimer correlation for a parallel configuration placed on the edges joining two unit cells of size 4 to be negligible throughout at low values of $\theta$ and to agree with the other parallel dimer configuration for $\theta > \theta_c$. 
If we invert the lattice, exchanging the roles of the ground states (implicitly mapping $\ket{1}\leftrightarrow\ket{3}$ and $\ket{2}\leftrightarrow\ket{4}$ using the notation defined in \eqref{unperturbedGroundState1} and \eqref{unperturbedGroundState4}) 
we find the resulting state, corresponding to the second parallel configuration joining the four site unit cells, being a ground state as well representing the symmetry broken two-fold degenerate ground state space in the weak coupling regime. 
This bias towards one of the two parallel dimer configurations can be explained from the blocking into a unit cell of 4 sites in the uMPS calculation. The simulation will favor the less correlated solution for this specific blocking being the parallel configuration $\ket{1}$ with trivial bond dimension for weak enough couplings.
Moreover, we find the same $\theta$ dependence of $\mcal{D}(1,2)$ and $\mcal{D}(0,3)$ and therefore show only $\mcal{D}(1,2)$ in Fig.~\ref{fig:dimerDimer}.
In the strong coupling regime $\theta>\theta_c$ we find the inverted state to be the same state as the original state in accordance with the symmetry present in the alternating dimer dimer correlators. We take this as evidence for a unique ground state corresponding to the spin one ground state.

Based on the numerical results laid out here, we therefore confirm that the system is 
in a two-fold degenerate symmetry broken phase where each of the symmetry broken states is in a trivial SPT phase with respect to the $SO(3)$ symmetry for couplings $\theta<\theta_c$ and undergoes a phase transition to the Haldane phase at $\theta=\theta_c\approx 0.18\pi$. In the regime of strong couplings, $\theta > \theta_c$ the essential physics of the BDC is therefore well captured by the effective spin-$1$ model.

\begin{figure*}
\includegraphics[width=\textwidth]{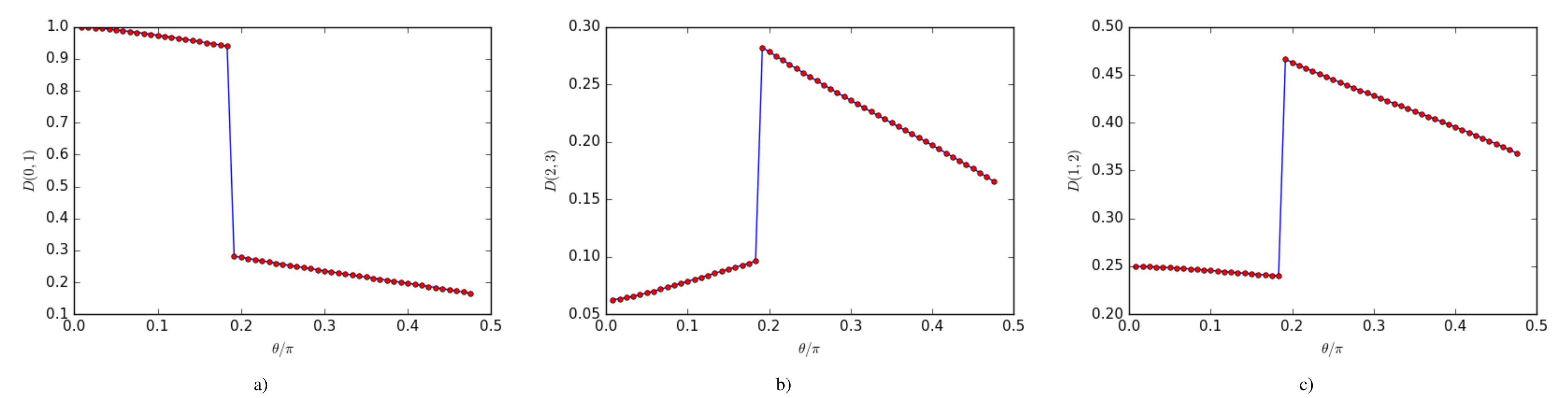}
\caption{Plot of the $\theta$ dependent dimer-dimer correlations along different pairs of anti-ferromagnetic edges. Figure a) on the left displays the dimer-dimer correlation $\mcal{D}(i,j)$ on edges $\langle0,2\rangle$ and $\langle1,3\rangle$ inside of one unit cell of size 4 used for the uMPS simulation. 
The middle figure b) shows the dimer-dimer correlator for dimers on the edges joining two unit cells of size 4 where the figure on the right display the correlator for the alternating configuration on the edges $\langle1,3\rangle$ and $\langle2,4\rangle$. One clearly can see the parallel configuration $\ket{1}$ being favored in the weak coupling regime as expected from to the weak coupling perturbation theory presented in Sec.~\ref{sec:Weak coupling limit}. In the strong coupling regime we find the alternating configuration to be favored giving rise to a symmetry protected long range entangled state similar to the AKLT-state.}
  \label{fig:dimerDimer}
\end{figure*}

\section{Conclusion}
In this work we have briefly summarised the formalism for classifying SPT phases in one-dimensional quantum lattice systems and with focus on the composition of quantum systems. 
Our discussion here focused on the case of the spin-$1/2$ models with effective spin-$1$ theories induced by variable ferromagnetic couplings between pairs of spins. 

We derived an order parameter similar to established order parameters for spin 1 systems that can be used in order to compare a composed spin $\frac{1}{2}\otimes\frac{1}{2}$ system with an approximate effective spin-$1$ model in a rigorous manner. In a pedagogical case study of the BDC we apply this parameter to show that in the strong coupling limit the treatment of the composed system as an effective spin-$1$ system is valid on the level of phases. 
We find that the weakly coupled system shows a different behavior from that of the spin-$1$ model already on the level of SPT phases which is in contrast to the intuition one might obtain from standard spin-$1/2$ Heisenberg ladders \cite{WhiteSpin1Chain, Pollmann2012, Dagotto1996, Totsuka1995, Matsumoto2004,PhysRevLett.105.077202}. 
In those ladders an arbitrarily weak coupling suffices to drive the system into the non-trivial Haldane phase.
We expect this difference to be rooted in the fact that a single $\Delta$-chain is degenerate but gapped as opposed to the Heisenberg chain \cite{Monti1991}. 
The transition of the BDC into the Haldane phase corresponding to the spin-$1$ physics occurs at the ferromagnetic coupling strength $J_F\approx0.66J_{AF}$ being of the same magnitude as the anti-ferromagnetic coupling. 
In general we assume that the coupling needs to be at least of the order of the energy scale introduced by the gap of the uncoupled system,
which in our case is of the order $\Delta E\approx 0.21 J_{AF}$\cite{Kubo1993}, 
in order to guarantee the spin-$1$ physics to be dominant in the local Hilbert space decomposition. At this strength the coupling can be strong enough to couple the unperturbed ground state space with the spin-$1$ sub-space suppressing the spin $0$ sub-space.

The presented scheme can be generalized given a model symmetric with respect to the symmetry group $G$ containing $\mathbb{Z}_2$ as a normal subgroup. The tensor product implicit in the blocking may lead to a new symmetry group $H\simeq G/\mathbb{Z}_2$. If the SPT phases of $H$ are different from those of $G$ an effective description by means of $H$ can be analyzed in the same fashion as we did concerning $G=SU(2)$ and $H\simeq SO(3)$. Moreover, considering blockings of $n$ sites one needs to investigate the difference between the second cohomology classes of $G$ and $H\simeq G/\mathbb{Z}_n$. 
The analysis presented here therefore  illustrates a way to numerically investigate the validity of an effective theory by testing their essential physical features based on the classification of SPT phases and also exemplifies how the composition and coupling of apparently trivial systems can lead to non-trivial symmetry protected topologically ordered phases.

\acknowledgments
We warmly thank Bella Lake for inspiring discussions that motivated some of this work. We thank Ioannis Rousochatzakis 
for early help with exact diagonalization tools. 
Moreover, we thank Marek Gluza, Andreas Bauer, and Zoltan Zimboras 
for fruitful discussions. This work was supported by the Helmholtz Centre Berlin,
the DFG (CRC 183, specifically project B1, and the Emmy Noether programme BE 5233/1-1), the
Wallenberg Academy Fellows program of the KAW Foundation,
the Templeton Foundation,
the ERC (TAQ), and the EU (AQuS). 

\section{Appendix}
\subsection{Local dynamics in the single-layered $\Delta$-chain}\label{append:localDynamics}
In this section, we briefly compute the local dynamics of excitations in a single-layered $\Delta$-chain. In particular, we give the elementary equations on four spins that define the dynamics on the full chain. Therefore, by $h_\Delta$ we denote the local Hamiltonian acting as $h_\Delta = \sum_{i=0}^2 s_i\cdot s_{i+1_{{\rm mod}3}}$. Then, we obtain
\begin{align*}
	&h_\Delta \ket{\nearrow,\nearrow} = -\frac{3}{4} \ket{\nearrow,\nearrow},\\
	&h_\Delta \ket{\nearrow,{{}_i}{\diagup}} = -\frac{3}{4} \ket{\nearrow,{{}_i}{\diagup}},\\
	&h_\Delta \ket{{{}_1}{\diagup},\nearrow} = 
	\frac{1}{4} \ket{{{}_1}{\diagup},\nearrow} + \frac{1}{2} \ket{{{}_1}{\diagup},{{}_0}{\diagup}}  - \frac{1}{2} \ket{{{}_0}{\diagup},{{}_1}{\diagup}}    ,\\
	&h_\Delta \ket{{{}_{-1}}{\diagup},\nearrow} = \frac{1}{4} \ket{{{}_{-1}}{\diagup},\nearrow} - \frac{1}{2} \ket{{{}_{-1}}{\diagup},{{}_0}{\diagup}}  + \frac{1}{2} \ket{{{}_0}{\diagup},{{}_{-1}}{\diagup}}   , \\
	&h_\Delta \ket{{{}_{0}}{\diagup},\nearrow} = \frac{1}{4} \ket{{{}_{0}}{\diagup},\nearrow} + \frac{1}{2} \ket{{{}_{1}}{\diagup},{{}_{-1}}{\diagup}}  - \frac{1}{2} \ket{{{}_{-1}}{\diagup},{{}_{1}}{\diagup}} ,\\   
	&h_\Delta \ket{{{}_{1}}{\diagup},{{}_{0}}{\diagup}} = \frac{1}{4} \ket{{{}_{1}}{\diagup},{{}_{0}}{\diagup}} + \frac{1}{2} \ket{{{}_{0}}{\diagup},{{}_{1}}{\diagup}}  + \frac{1}{2} \ket{{{}_{1}}{\diagup},\nearrow} ,\\
	&h_\Delta \ket{{{}_{-1}}{\diagup},{{}_{0}}{\diagup}} = \frac{1}{4} \ket{{{}_{-1}}{\diagup},{{}_{0}}{\diagup}} + \frac{1}{2} \ket{{{}_{0}}{\diagup},{{}_{-1}}{\diagup}}  - \frac{1}{2} \ket{{{}_{-1}}{\diagup},\nearrow}, \\   
	&h_\Delta \ket{{{}_{0}}{\diagup},{{}_{0}}{\diagup}} = \frac{1}{4} \ket{{{}_{0}}{\diagup},{{}_{0}}{\diagup}} + \frac{1}{2} \ket{{{}_{1}}{\diagup},{{}_{-1}}{\diagup}}  + \frac{1}{2} \ket{{{}_{-1}}{\diagup},{{}_{1}}{\diagup}} ,\\   
	&h_\Delta \ket{{{}_{0}}{\diagup},{{}_{1}}{\diagup}} = \frac{1}{4} \ket{{{}_{0}}{\diagup},{{}_{1}}{\diagup}} + \frac{1}{2} \ket{{{}_{1}}{\diagup},{{}_{0}}{\diagup}}  - \frac{1}{2} \ket{{{}_{1}}{\diagup},\nearrow} ,\\   
	&h_\Delta \ket{{{}_{0}}{\diagup},{{}_{-1}}{\diagup}} = \frac{1}{4} \ket{{{}_{0}}{\diagup},{{}_{-1}}{\diagup}} + \frac{1}{2} \ket{{{}_{-1}}{\diagup},{{}_{0}}{\diagup}}  - \frac{1}{2} \ket{{{}_{-1}}{\diagup},\nearrow}, \\   
	&h_\Delta \ket{{{}_{1}}{\diagup},{{}_{-1}}{\diagup}} = \frac{1}{4} \ket{{{}_{1}}{\diagup},{{}_{-1}}{\diagup}} + \frac{1}{2} \ket{{{}_{0}}{\diagup},{{}_{0}}{\diagup}}  - \frac{1}{2} \ket{{{}_{0}}{\diagup},\nearrow}, \\   
	&h_\Delta \ket{{{}_{-1}}{\diagup},{{}_{1}}{\diagup}} = \frac{1}{4} \ket{{{}_{-1}}{\diagup},{{}_{1}}{\diagup}} + \frac{1}{2} \ket{{{}_{0}}{\diagup},{{}_{0}}{\diagup}}  - \frac{1}{2} \ket{{{}_{0}}{\diagup},\nearrow} .
\end{align*}
From this it follows, in particular, that the local Hamiltonian terms $h_\Delta$ may propagate the triplets through the chain, whilst the triangle initially occupied by the excitation $\ket{{}_i{\diagup}}$ will always be covered by some triplet excitation.


\bibliographystyle{apsrev4-1}
%

\end{document}